\documentclass[aps,prl,reprint,superscriptaddress]{revtex4-2}
\usepackage[T1]{fontenc}
\usepackage{amsmath}
\usepackage{amsthm}
\usepackage{amssymb}
\usepackage{graphicx}
\usepackage{qcircuit}
\usepackage{braket}
\usepackage{comment}
\usepackage{bbm}
\usepackage{float}
\usepackage{bm}
\usepackage{dcolumn}
\usepackage{physics}
\usepackage{subcaption}
\usepackage{svg}
\usepackage{caption}
\usepackage{hyperref}
\usepackage{color}

\hypersetup{colorlinks,linkcolor={blue},citecolor={magenta},urlcolor={blue}}

\newcommand{\A}{\h{A}}
\newcommand{\beq}{\begin{equation}}
\newcommand{\eeq}{\end{equation}}
\definecolor{Pr}{rgb}{0.4,0.3,0.9}

\newcommand{\spliteq}[1]{\begin{equation}
\begin{split}
#1
\end{split}
\end{equation}
}

\def\ket#1{\mathop{|#1\rangle}}
\def\bra#1{\mathop{\langle #1|}}

\def\la{\lambda}

\definecolor{JM}{RGB}{4,116,149}

\def\h{\hat}
\def\U{\hat{U}}
\def\p{\hat{p}}
\def\F{\hat{F}}
\def\M{\hat{M}}
\def\x{\hat{x}}
\def\X{\hat{X}}
\def\H{\hat{H}}
\def\T{^{T}}
\theoremstyle{remark}
\newtheorem*{claim}{Claim}
\theoremstyle{remark}
\newtheorem*{note}{Note}

\usepackage{comment}
\usepackage{graphicx}
\usepackage{amsmath}
\usepackage{amssymb}
\usepackage{changepage}
\usepackage{mathtools}
\usepackage[thinc]{esdiff}
\usepackage{multirow}
\usepackage{braket}
\PassOptionsToPackage{margin=1in}{geometry}
\usepackage{lipsum}
\usepackage{setspace}

\usepackage{listings}
\usepackage{xcolor}
\usepackage[ruled,lined]{algorithm2e}
\usepackage{algpseudocode}

\definecolor{codegreen}{rgb}{0,0.6,0}
\definecolor{codegray}{rgb}{0.5,0.5,0.5}
\definecolor{codepurple}{rgb}{0.58,0,0.82}
\definecolor{backcolour}{rgb}{0.95,0.95,0.92}

\lstdefinestyle{mystyle}{
    backgroundcolor=\color{backcolour},   
    commentstyle=\color{codegreen},
    keywordstyle=\color{magenta},
    numberstyle=\tiny\color{codegray},
    stringstyle=\color{codepurple},
    basicstyle=\ttfamily\footnotesize,
    breakatwhitespace=false,         
    breaklines=true,                 
    captionpos=b,                    
    keepspaces=true,                 
    numbers=left,                    
    numbersep=5pt,                  
    showspaces=false,                
    showstringspaces=false,
    showtabs=false,                  
    tabsize=2
}
\begin{document}
\title{Quantum Dynamical Hamiltonian Monte Carlo}

\affiliation{Department of Computer Science, Rensselaer Polytechnic Institute, Troy, NY 12180, USA}
\affiliation{Extropic Corp., San Francisco, CA 94111, USA}
\affiliation{Dirac Inc., New York, NY 10001, USA}
\affiliation{Institute for Quantum Computing, University of Waterloo, Waterloo, Ontario, N2L 3G1, Canada}
\affiliation{Department of Applied Mathematics, University of Waterloo, Waterloo, Ontario, N2L 3G1, Canada}

\author{Owen Lockwood}
\email{owen@extropic.ai}
\affiliation{Department of Computer Science, Rensselaer Polytechnic Institute, Troy, NY 12180, USA}
\affiliation{Extropic Corp., San Francisco, CA 94111, USA}

\author{Peter Weiss}
\affiliation{Dirac Inc., New York, NY 10001, USA}

\author{Filip Aronshtein}
\affiliation{Dirac Inc., New York, NY 10001, USA}

\author{Guillaume Verdon}
\email{gv@extropic.ai}
\affiliation{Institute for Quantum Computing, University of Waterloo, Waterloo, Ontario, N2L 3G1, Canada}
\affiliation{Department of Applied Mathematics, University of Waterloo, Waterloo, Ontario, N2L 3G1, Canada}
\affiliation{Extropic Corp., San Francisco, CA 94111, USA}

\begin{abstract}

One of the open challenges in quantum computing is to find meaningful and practical methods to leverage quantum computation to accelerate classical machine learning workflows. 
A ubiquitous problem in machine learning workflows is sampling from probability distributions that we only have access to via their log probability. To this end, we extend the well-known Hamiltonian Monte Carlo (HMC) method for Markov Chain Monte Carlo (MCMC) sampling to leverage quantum computation in a hybrid manner as a proposal function. Our new algorithm, Quantum Dynamical Hamiltonian Monte Carlo (QD-HMC), replaces the classical symplectic integration proposal step with simulations of quantum-coherent continuous-space dynamics on digital or analogue quantum computers. We show that QD-HMC maintains key characteristics of HMC, such as maintaining the detailed balanced condition with momentum inversion, while also having the potential for polynomial speedups over its classical counterpart in certain scenarios.  As sampling is a core subroutine in many forms of probabilistic inference, and MCMC in continuously-parameterized spaces covers a large-class of potential applications, this work widens the areas of applicability of quantum devices. 

\end{abstract}
\maketitle

\section{Introduction}

Fueled by the success of machine learning (ML) \cite{lecun2015deep} and recent developments in quantum computing hardware \cite{kim2023evidence, arute2019quantum, zhong2020quantum, chow2021ibm, wu2021strong, madsen2022quantum}, substantial interest has developed at the intersection of these fields \cite{biamonte2017quantum, preskill2018quantum, bharti2021noisy}. Sampling from difficult distributions is key to many classical machine learning workflows. Sampling routines are often leveraged as a critical component of workloads such as Bayesian inference \cite{von2011bayesian}, optimization \cite{bertsimas1993simulated}, machine learning \cite{porteous2008fast}, statistical inference/modelling \cite{van2021bayesian}, and Energy Based Models (EBMs) \cite{lecun2006tutorial, du2019implicit} to name a few. Research at the intersection of classical and near term quantum machine learning has been heavily focused on parameterized quantum circuits \cite{benedetti2019parameterized} for problems such as quantum simulation \cite{yuan2019theory,endo2020variational,sbahi2022provably}, reinforcement learning \cite{chen2020variational,lockwood2020reinforcement,lockwood2021playing}, and mathematical applications \cite{anschuetz2019variational,kubo2021variational,lubasch2020variational}; however, we explore a different direction to accelerate classical machine learning workflows. Specifically, we investigate the potential for quantum computers to accelerate Hamiltonian Monte Carlo (HMC) \cite{duane1987hybrid} proposals. HMC is a continuous parameter space Markov Chain Monte Carlo (MCMC) method that approximates hamiltonian dynamics to provide improved proposals \cite{betancourt2017conceptual, beskos2013optimal}. MCMC methods \cite{robert2011short}, such as HMC, are some of the most established methods in classical machine learning and provide a general purpose toolkit for sampling from target probability distributions. Accelerating MCMC with quantum proposals is a research direction that has only just begun \cite{mazzola2021sampling,layden2022quantum, nakano2023qaoa, orfi2023near, mansky2023sampling, mazzola2024quantum, ferguson2024quantum}.

As novel physics-based accelerators for probabilistic sampling are on the horizon, it is important to benchmark the performance of quantum computers for this task. In this paper we take a first step towards porting a core sampling algorithm to quantum devices. Here, we propose a method for extending the HMC approach to leverage quantum computation via an algorithm we call Quantum Dynamical Hamiltonian Monte Carlo (QD-HMC). Our algorithm directly builds upon the work of \citeauthor{layden2022quantum} \cite{layden2022quantum}, in which they proposed the use of quantum simulations to generate MCMC proposals in discrete state spaces, and \citeauthor{verdon2019quantum} \cite{verdon2019quantum}, which proposed a Continuous Variable (CV) QAOA \cite{farhi2014quantum}. CV-QAOA is effectively a variational Trotterization of continuous space dynamics, similar to QD-HMC's randomized Trotterization of dynamics. Our method utilizes quantum computers to more efficiently simulate the Hamiltonian dynamics used for the proposal for HMC. We outline the theory behind this algorithm and present initial small simulations of QD-HMC. Given the classical complexity of simulating quantum dynamics, there is clear potential in using quantum hardware to generate proposals. QD-HMC leverages hybrid quantum computing to effectively create a method for quantum devices to be relevant to probabilistic and Bayesian inference at large scale. We hope that this link will fuel further explorations of scalability and performance for QD-HMC.

\section{Background}

There are many cases in physics, machine learning, and optimization, in which we do not have access to a desired target distribution, but we do have access to an energy function (or an oracle for the un-normalized negative log-likelihood) \cite{huembeli2022physics}. To sample from this target distribution we can use MCMC methods, which often requires only the ability to evaluate the log-probability (i.e.\ the energy). Mathematically, we can represent these distributions as a Boltzmann distributions where $P(x) = \frac{1}{\mathcal{Z}} e^{-E(x)/kT}$ with $\mathcal{Z}$ being the partition function $\int e^{-E(x)/kT} dx$. In many cases, such as that of modern deep neural network based EBMs, $kT$ is just set to $1$. One of the earliest and most popular MCMC sampling methods is the Metropolis-Hastings algorithm \cite{hastings}, which achieves sampling from the target distribution by repeatedly alternating between proposal and acceptance steps. In Metropolis-Hasting, a proposal $q(y | x)$ is generated using any one of many possible methods, e.g. random walk, Langevin dynamics \cite{roberts1996exponential, cheng2018convergence}, splines \cite{shao2013efficient}, or normalizing flows \cite{brofos2022adaptation}, which is then fed into the acceptance function $A(y | x) = \min \left[ 1, \frac{\pi (y) q(y|x)}{\pi (x) q(x|y)} \right ] $. In our notation, $\pi(x) \equiv P(x)$, which represents the target distribution. In many cases, the proposal function is symmetric (i.e. $q(x|y) = q(y|x)$), which eliminates the need for the so called ``Hasting's Correction'' and reduces $A(y|x)$ to $\min \left[ 1, \frac{\pi (y)}{\pi (x)} \right ] $. In the case of energy based models, this becomes $\min \left[ 1, e^{E(x) - E(y)} \right ] $.

In continuous high dimensional state spaces, HMC provides a way to leverage Hamiltonian Dynamics and the gradient information of the log probability to simulate the trajectories of particles traversing the energy landscape \cite{betancourt2017geometric}. Specifically, HMC adds an auxiliary variable momentum $p$ to create a kinetic energy term (in addition to the log probability term defined to be the potential energy). Given the function $\pi(x)$, HMC draws from a join density $\pi(x, p) = \pi(p | x) \pi(x)$. The Hamiltonian of this joint density is $H = - \log \pi(x, p) = -\log \pi(p | x) - \log \pi(x) = K(x, p) + U(x)$ (i.e. the sum of kinetic and potential energies). By picking a value for the momentum, one can simulate the Hamiltonian dynamics $\dv{x}{t} = \pdv{H}{p} = -\pdv{p} \log \pi(p|x)$ and $\dv{p}{t} = -\pdv{H}{x} = \pdv{}{x} \log \pi(x)$. Although there are a variety of methods for integrating approximations for these dynamics \cite{bou2018geometric}, one of the most common methods is a leapfrog integrator that approximates these dynamics via the following finite difference estimation

\spliteq{\label{eq:leap}
    p_{i + 1} &= p_i - \frac{\epsilon}{2} \pdv{}{x} \log \pi(x_i) \\
    x_{i + 1} &= x_i + \epsilon M^{-1} p_{i+1} \\
    p_{i + 1} &= p_{i + 1} - \frac{\epsilon}{2} \pdv{}{x} \log \pi(x_{i+1})
}

where $\epsilon$ is the step size and $M$ is the mass matrix (often a multiple of the identity matrix) \cite{neal2011mcmc}. This proposal is then fed into the acceptance function $\min \left [1, \exp \left ( -H(x_{i + 1}, p_{i + 1}) + H(x_i, p_i) \right ) \right ]$.

Since HMC proposals (ideally) conserve energy \cite{bou2018geometric}, their acceptance probability can be much higher than other MCMC methods. Naturally, if we could perfectly simulate Hamiltonian dynamics the acceptance probability would be 1 \cite{duane1987hybrid}. Since the Hamiltonian dynamics are energy-conserving, the acceptance probability would always be $e^0 = 1$. Our aim is to show how these classical dynamics (classically simulated via symplectic leapfrog integration) can be replaced using a quantum computational simulation of dynamics, and how this confers a potential quantum advantage for the task of sampling from distributions in continuous spaces.

\begin{figure*}
    \centering
    \includegraphics[width=0.95\linewidth]{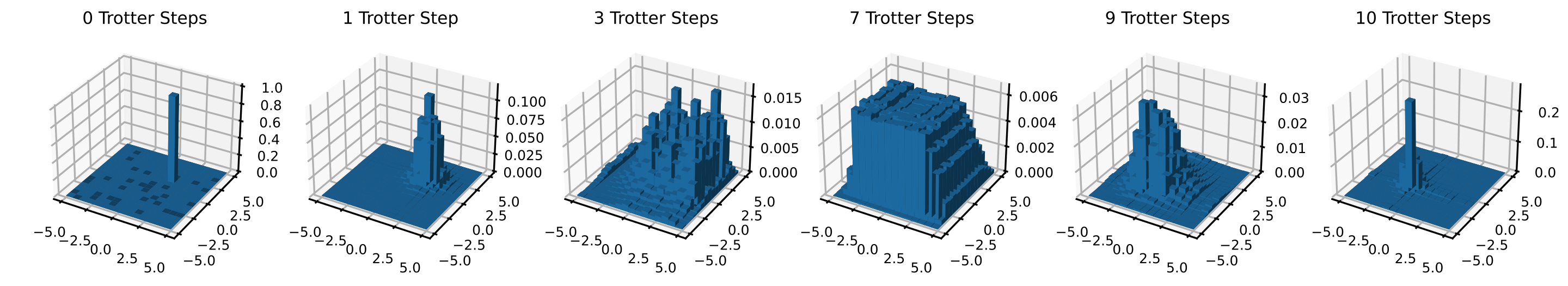}
    \caption{A QD-HMC update: 10 Trotter Steps showing the evolution of the Wavefunction probability for a 2D Gaussian $\mu=(0,0)$}
    \label{fig:wf}
\end{figure*}

\section{Quantum Dynamical Hamiltonian Monte Carlo}

Quantum Dynamical Hamiltonian Monte Carlo (QD-HMC) consists of leveraging digitally-simulated or analogue continuous quantum dynamics in order to suggest proposals for a Metropolis-Hastings acceptance step in continuous spaces. Let us first outline the steps of the QD-HMC algorithm before demonstrating its theoretical and empirical justifications. This can also be seen in Algorithm \ref{alg:app_algo}. We present the algorithm for digitally simulated quantum dynamics here.

\textbf{\textit{QD-HMC Algorithm---}}
First, prepare bitstring $\ket{x}$. With $f(\x)$ as the target distribution, define $\H_{\eta\la} = \eta\frac{\p^{2}}{2} + \la f(\x)$ with $\p$ being the momentum and chosen hyperparameters $\eta, \la$. Via random Trotterization, apply $\U_{\eta\la} \equiv e^{-i \H_{\eta\la} t}$. Sample bitstring measurement $y \sim \lvert\bra{y}\U_{\eta\la}\ket{x}\rvert^{2}$. Accept with probability $A(y|x)=\min\{1,e^{[-f(x)+f(y)]}\}$.

To understand how QD-HMC relates to traditional HMC, we need to look at the ansatz and how each layer in the simulated evolution maps the position and momenta in the Heisenberg picture. Our algorithm can be seen as a merging of Quantum Enhanced MCMC \cite{layden2022quantum} with HMC using the continuous evolution techniques of \citeauthor{verdon2019quantum} \cite{verdon2019quantum}. The evolution of these Trotter updates are very similar to those of a CV-QAOA, a connection readily seen in the following section.

Suppose the log probability is a function of $N$ variables; $f(\bm{x}): \mathbb{R}^N \mapsto \mathbb{R}$. As a quantum operator, this function becomes part of the Hamiltonian $f(\bm{\hat{x}})$, where $\bm{\hat{x}}$ is the position operator corresponding to a quadrature of a simulated $N$-dimensional quantum system. In photonic or analogue systems, such quadrature operators could be represented directly, while in digital quantum systems, this quadrature can be represented in the one-hot or binary representation. In this work we follow the convention of \cite{somma2015quantum} and \cite{verdon2018universal}. That is to say, we represent a discretized position operator $\hat{x}_d$ over $d$ qubits via:
\begin{equation}\label{eq:xd}
	\hat{x}_d = \sqrt{\tfrac{2\pi}{N}}\big( -\hat{J} + \left(\tfrac{N}{2}-1\right) I \big)
\end{equation}
where $N=2^d$ and $\hat{J}_d\equiv\sum_{j=0}^{d-1} 2^{(d-1)-j} \ket{1}\bra{1}^{(j)}$ with $\ket{1}\bra{1}^{(j)}$ acting on the $j^\text{th}$ qubit only. A computational basis state $\ket{k}$ on $d$ qubits is an eigenvector of $\hat{x}_d$ with the eigenvalue $x_k = \sqrt{2\pi/N}\left(k -N/2 \right)$. The states $\ket{k}$ therefore represent positions on a one-dimensional grid with half-open interval $\sqrt{2\pi}[-\sqrt{N}/2, \sqrt{N}/2)$, and the binary representation of $k$ provides a little-endian description of position with respect to the negative bound of the domain. Using a Centered Fourier Transform $\hat{F}_c$ \cite{somma2015quantum}, defined as $\hat{F}_c = \hat{X}_0 \hat{F} \hat{X}_0$ with $\hat{X}_0$ being the Pauli X gate on the 0th qubit and $\hat{F}$ being the Quantum Fourier Transform \cite{coppersmith2002approximate}, the momentum operator in this space is defined as $\hat{p}_d = \hat{F}_c \hat{x}_d \hat{F}_c^\dagger$ and has corresponding momentum eigenstates $\ket{\rho_k} = \hat{F}\ket{k}$ with eigenvalues identical to the position eigenvalues.

Now that we have outlined how to represent the function of interest, we can focus on understanding the Hamiltonian $\hat{H}_{\eta \lambda}$. First, let us focus on what is called the kinetic term  $\hat{K}=\tfrac{1}{2}\sum_{j=1}^N\hat{p}_j^2:=\tfrac{1}{2}\bm{\hat{p}}^2$ \cite{verdon2019quantum}. For analogue devices, this represents the momentum, or the canonically conjugate variable to the position where we have the commutation relation $[\hat{x}_j,\hat{p}_k] = i\delta_{jk}$. For the binary representation of these operators, we can simply use the fact that momentum is analogous to position in Fourier space. Using a Centered Fourier Transform $\hat{F}_c$ \cite{somma2015quantum}, the momentum operator in this space is thus defined as $\hat{p}_d = \hat{F}_c^\dagger \hat{x}_d \hat{F}_c$.

Next, let us examine the action of each layer in the random Trotterization of the Hamiltonian. In the Heisenberg picture, the kinetic term of the Hamiltonian generates a change in position of the form

\begin{equation}\label{eq:parshift}
    e^{i\eta\bm{\hat{p}}^2/2}   \, \bm{\hat{x}} \,
    e^{-i \eta\bm{\hat{p}}^2/2} =\bm{\hat{x}} + \eta \bm{\hat{p}}
\end{equation}

where $\eta$ is a choice of parameter analogous to inverse mass, $\eta \cong m^{-1}$. With this, we can see that the position gets updated by the momentum divided by mass (i.e. the velocity). Now, for the next step, the momentum operator is translated through evolution under the target Hamiltonian as

\begin{equation}\label{eq:momshift}
    e^{i\lambda f(\bm{\hat{x}})}    \bm{\hat{p}}e^{-i\lambda f(\bm{\hat{x}})} = \bm{\hat{p}}- \lambda\, \nabla \! f(\bm{\hat{x}}) 
\end{equation}

thus the momentum is shifted in a manner proportional to the negative gradient of the target function. The details of these derivations are presented in Appendix A. Evolving under both the target and kinetic Hamiltonians yields

\spliteq{\label{eq:double_update}
       \bm{\hat{x}}&\rightarrow \bm{\hat{x}} +\eta\bm{\hat{p}}- \eta\lambda\, \nabla \! f(\bm{\hat{x}})
}
This is analogous to gradient descent with momentum, or classical kinematics, with $\eta = \Delta t/m$ and $\lambda =\Delta t$ for time step size $\Delta t$. For infinitesimal time steps, alternating between Equation \ref{eq:parshift} and Equation \ref{eq:momshift} two steps is equivalent to the quantum dynamics of a particle undergoing motion in a high-dimensional potential, i.e. kinematic evolution.

As an illustrative example of the dynamics of the Hamiltonian simulation done by the QD-HMC algorithm, see Figure \ref{fig:wf}. For a log probability, or target function, that is a 2D Gaussian with mean at $(0,0)$ we show the probability of the wavefunction over 10 trotter steps. The wavefunction is initialized to a location that is away from the mean and we can see the initial wavefunction is 100\% in the location because every step of QD-HMC is initialized to the single state $|x\rangle$. As the evolution progresses, we see the final result is a wavefunction localized at the desired mean $(0,0)$. This is not steps of HMC, but Trotter steps within a single run of our quantum proposal function. Although this evolution is heavily dependent on the number of trotter steps and the time simulated by these steps ($t = 2$ in this example), this provides intuition and pedagogical insight into how the algorithm proposal step works.

Let us highlight the connection between the above dynamics and the symplectic integration more explicitly. Taking Equation \ref{eq:leap} and looking only at the position update, we can see the leapfrog integrator's x update is
\begin{equation}\label{eq:rewrite}
    x \rightarrow x + \epsilon M^{-1} \left ( p - \frac{\epsilon}{2} \nabla \log \pi(x) \right )
\end{equation}
With $\pi(x) = e^{f(x)}$ we can see this reduces to 
\begin{equation}\label{eq:sym_comp}
    x \rightarrow x + \epsilon M^{-1} p - \frac{\epsilon^2}{2} M^{-1} \nabla f(x) 
\end{equation}
Since the inverse of the diagonal mass matrix is the diagonal matrix of inverses, $M^{-1}_{i,i} = (M_{i,i})^{-1}$, and we have established $\eta \approx \frac{\Delta t}{m}$ then we can see these terms are equivalent with $\epsilon = \Delta t$. With $\eta = \frac{\epsilon}{m}$, we can set $\lambda =\frac{\Delta t}{2} = \frac{\epsilon}{2}$. We can now see clearly the similarities of the hamiltonian dynamics between symplectic Leapfrog integration in Equation \ref{eq:sym_comp} and the Quantum Hamiltonian Dynamics in Equation \ref{eq:double_update}. We see that we recover the symplectic integrator behavior, except this is done quantum coherently, allowing for exponentially complex superpositions over position values to be simultaneously symplectically integrated.

Now that we have illustrated the mechanisms of Hamiltonian simulation, we will outline how QD-HMC functions as a Markov Chain Monte Carlo (MCMC) algorithm. First, we show that our QD-HMC algorithm meets the requirement of conserved energy. In the classical case this can be easily shown via $\dv{H}{t} = \sum_i \left [\pdv{H}{x_i} \dv{x_i}{t} + \pdv{H}{p_i} \dv{p_i}{t} \right ] = \sum_i \left [\pdv{H}{x_i} \pdv{H}{p_i} - \pdv{H}{p_i} \pdv{H}{x_i} \right ] = 0 $ \cite{neal2011mcmc}. In the quantum case, we can use inspiration from the Ehrenfest theorem to show our quantum Hamiltonian conserves energy. As was shown in Appendix A, $\dv{\hat{x}}{t} = \eta \hat{p} = \pdv{\hat{H}}{\hat{p}}$ and $\dv{\hat{p}}{t} = - \lambda \nabla f(\hat{x}) = - \pdv{\hat{H}}{\hat{x}}$. Thus, following these derivations  $\dv{\hat{H}}{t} = \sum_i \left [\pdv{\hat{H}}{x_i} \dv{x_i}{t} + \pdv{\hat{H}}{p_i} \dv{p_i}{t} \right ] = 0 $ since this becomes identical to the classical case. Using these same derivatives, it is straightforward to show the quantum Hamiltonian dynamics preserve volume since the divergence of the vector field $\nabla \cdot \mathbf{F}$, $\mathbf{F} = \left ( \hat{x}, \hat{p} \right ) \rightarrow \left ( \pdv{\hat{H}}{\hat{p}}, - \pdv{\hat{H}}{\hat{x}} \right )$ is 0 (this is the same approach as in \cite{neal2011mcmc}). 

Next, a key requirement of MCMC algorithms is reversibility, i.e. it must satisfy the detailed balance condition. This condition states that for a stationary distribution $\pi$ and transition probability $p_{i, j} = P(x_{t + 1} = j | x_t = i)$
\begin{equation}
    \pi_i p_{i,j} = \pi_j p_{j, i} \; \; \;  \forall i, j\ .
\end{equation}
We show that QD-HMC obeys this condition by proving the symmetry of the proposal function $| \langle y | \hat{U}_{\eta \lambda} | x \rangle |^2 = | \langle x | \hat{U}_{\eta \lambda} | y \rangle |^2$ for our algorithm. The proof is presented in Appendix B. Since our proposal is symmetric, we can augment it with a classical acceptance/rejection step which meets the detailed balance condition. In HMC, one needs to add a sign flip to the momentum term to ensure reversibility. Intuitively, the Hamiltonian dynamics only go ``forward'', so $q(x|y) \rightarrow 0$, which means proposals would never be accepted. Adding a sign flip on the momentum, $(x, p) \rightarrow (x, -p)$ enable this reversibility (but is of no practical importance since the momentum is resampled), see \cite[sec 5.4]{betancourt2017conceptual} for more information. Analogous to this, there is a momentum flip necessary at the end of our quantum proposal to ensure reversibility, but as with HMC is of no practical relevance due to the random trotterization. Having established the theory of QD-HMC, we now present some example experiments to outline how the algorithm functions. 

\begin{figure}
    \centering
    \includegraphics[width=0.95\linewidth]{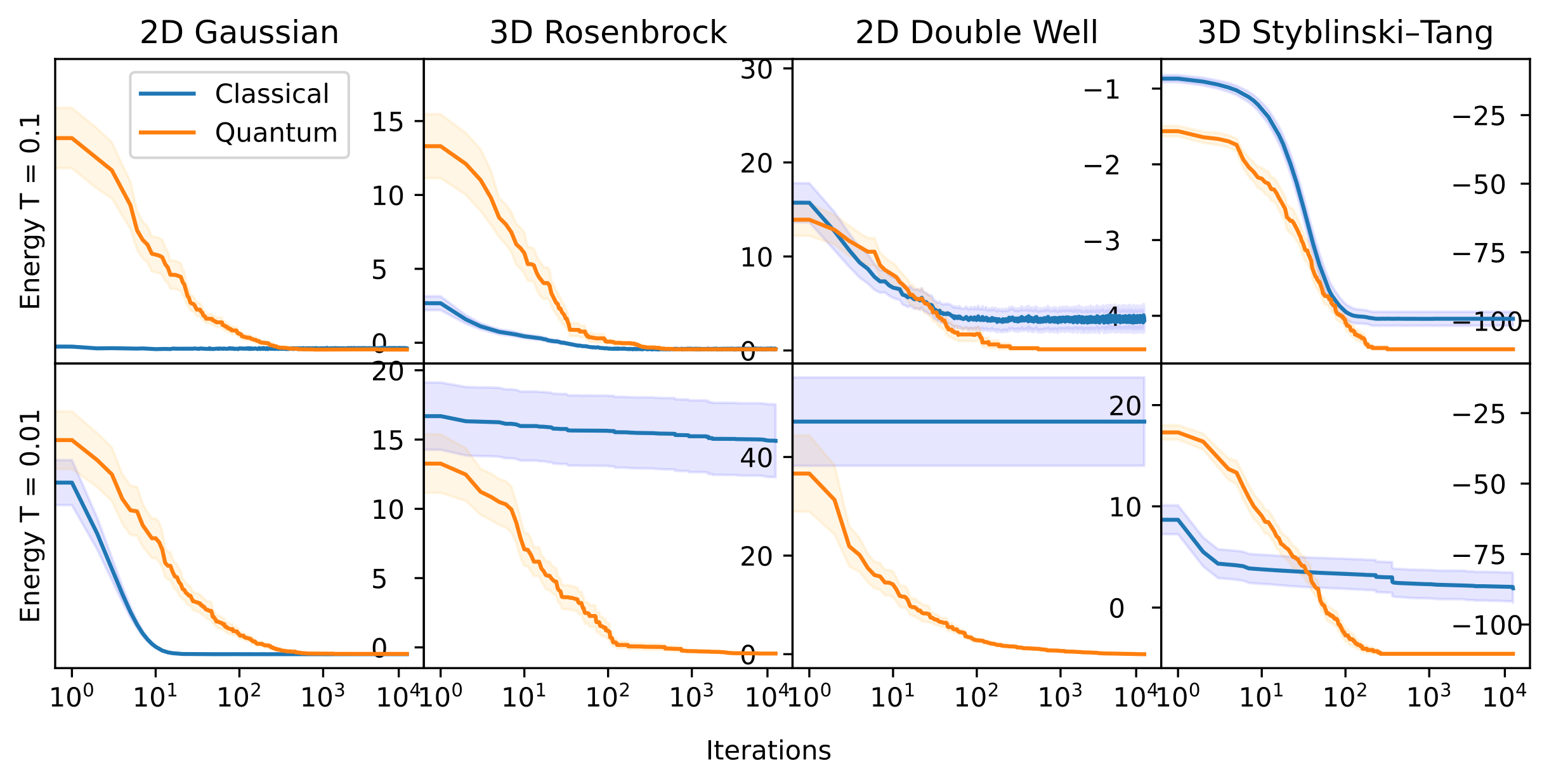}
    \caption{Comparison of QD-HMC and HMC samples energies over 10,000 iterations on a variety of optimization problems for temperature 0.1 and 0.01}
    \label{fig:comp}
\end{figure} 

\section{Simulations \& Experiments}

We evaluate our algorithm on a number test functions \cite{yang2010test} as is common in optimization literature \cite{beiranvand2017best, more2009benchmarking}. The QD-HMC and classical comparisons are generated from implementations built on TensorFlow Probability \cite{dillon2017tensorflow}, TensorFlow Quantum \cite{tfq} and Continuous Variable TensorFlow Quantum \footnote{\href{https://github.com/QuantumVerd/cv-tfq}{github.com/QuantumVerd/cv-tfq}}. All code is available at \footnote{\href{https://github.com/diracq/qdhmc}{github.com/diracq/qdhmc}}. Note that all results come from a noiseless and exact statevector simulations. The results of these experiments can be seen in Figure \ref{fig:comp}, which compares the energy over $10,000$ proposals on a set of functions for low temperatures $T = 0.1$ and $0.01$. We add this temperature so the adjusted log prob becomes $\log p(x) / T$ (i.e. the same way temperature is included in the Boltzmann distribution). The x axis is scaled logarithmically. For each function, we used Optuna \cite{akiba2019optuna} to optimize the classical and quantum hyperparameters independently. It is important to emphasize that these empirical results are not meant to suggest that QD-HMC is universally better than HMC, or that QD-HMC should be used (with simulations) as a replacement. These results are purely a proof of concept and of intuitive, explanatory, and pedagogical interest. There are many HMC improvements that could likely do even better on these problems \cite{hoffman2014no,chen2014stochastic}. The log probability for each function are defined as followed, Gaussian: $\sum_i -x_i-x_i^2$, Rosenbrock: $-\sum_i 10 (x_{i+1} - x_{i})^2 + (1 - x_i)^2$, Double Well: $-(x_0^4 - 4 x_0^2 + x_1^2) - 0.5x_0$, and Styblinski-Tang: $-\frac{1}{2} \sum_i x_i^4 - 16x_i^2 + 5x_i$. These show the minimization of the negative log probability (energy) function (labelled at the top) for low temperatures, making this energy a good proxy for free energy. Although the results are comparable at higher temperatures with HMC performing very well, the low temperature results highlight the potential of quantum dynamical simulation. The initial points for optimization are selected randomly across the range of possible digital quantum representations (which is often far from the minimum, something, as a qualitative point, HMC sometimes struggles with overcoming). These results are averaged over 50 repetitions with different initial points.

\begin{figure}
    \centering
    \includegraphics[width=0.95\linewidth]{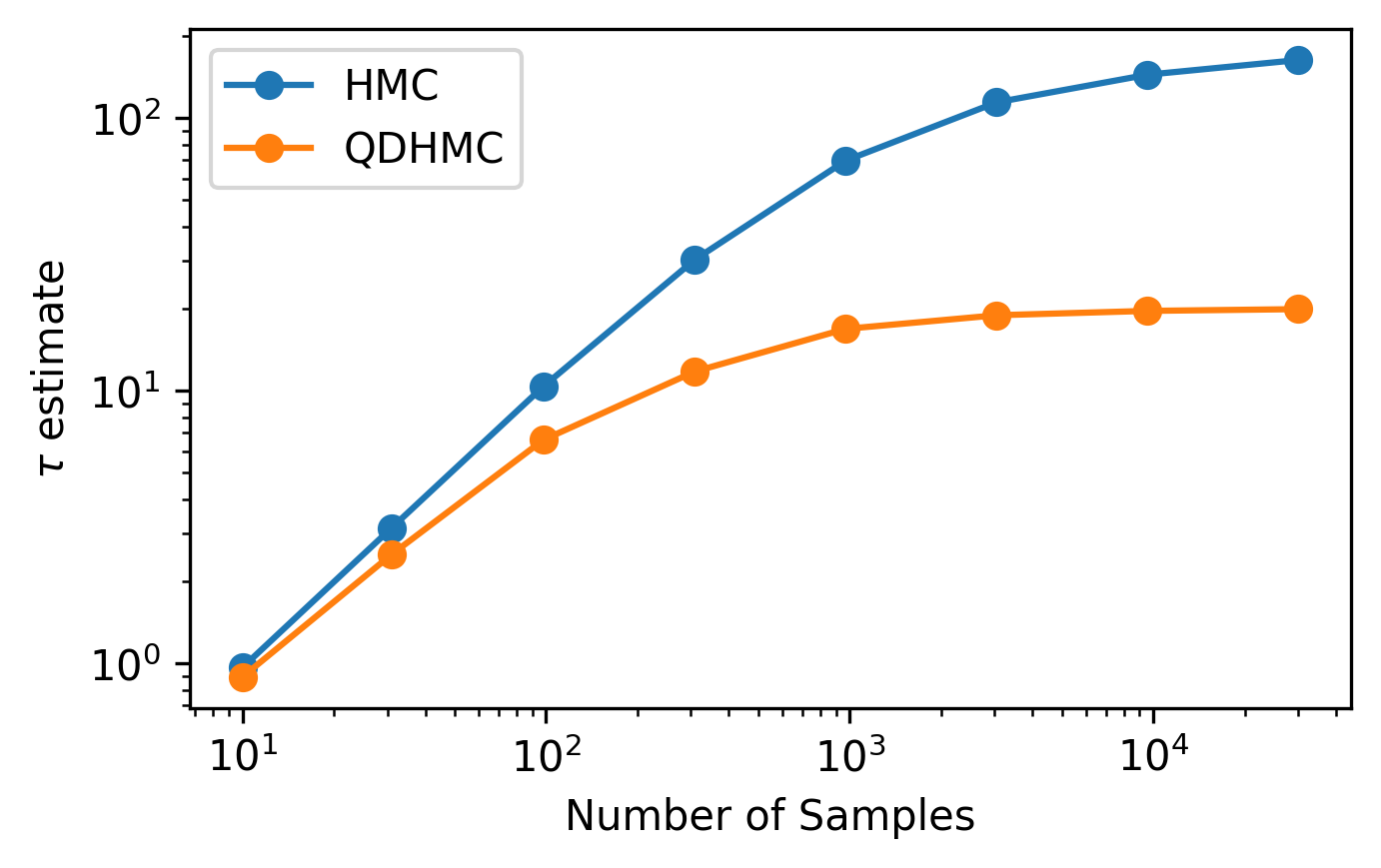}
    \caption{Estimation of autocorrelation time $\tau$ as a function of number of samples for a double well}
    \label{fig:tau}
\end{figure} 

To further understand the advantages QDHMC might yield, we also estimate the autocorrelation time across these problems for QDHMC and HMC. The integrated autocorrelation time $\tau$ is estimated for an finite chain observable $\{f_i\}_{i=1}^{N}$ via \cite{sokal1997monte, foreman2013emcee}

\spliteq{\label{eq:ac}
    \tau &= \frac{1}{2} + \sum_{t=1}^{\infty} \rho_f (t) \\
    \rho_f (t) &= \frac{N \sum_{i=1}^{N - t} (f_i - \mu_f) (f_{i+t} - \mu_f)}{(N - t)\sum_{i=1}^N (f_i - \mu_f)^2}
}

We plot an example of the estimation of autocorrelation time as a function of samples in Figure \ref{fig:tau}. The plots for the other problems are available in Appendix C. These experiments were all performed (and averaged) over the 2D version of functions and were conducted at $T = 5.0$. As these experiments reveal, QDHMC is able to achieve lower autocorrelation times, which means samples become independent faster and can result in improved sampling convergence rates. Although there are a multitude of methods when it comes to evaluating MCMC samplers, autocorrelation time is a common and important foundation. For example, it forms the backbone of effective sample size (ESS) estimation (ESS is proportional to $\frac{1}{\tau}$). As before, these results are not meant to offer strong and universal claims of quantum advantage but to investigate regimes of interest and mechanisms of potential advantage (such as having a lower autocorrelation time).

\begin{figure}
    \centering
    \includegraphics[width=0.95\linewidth]{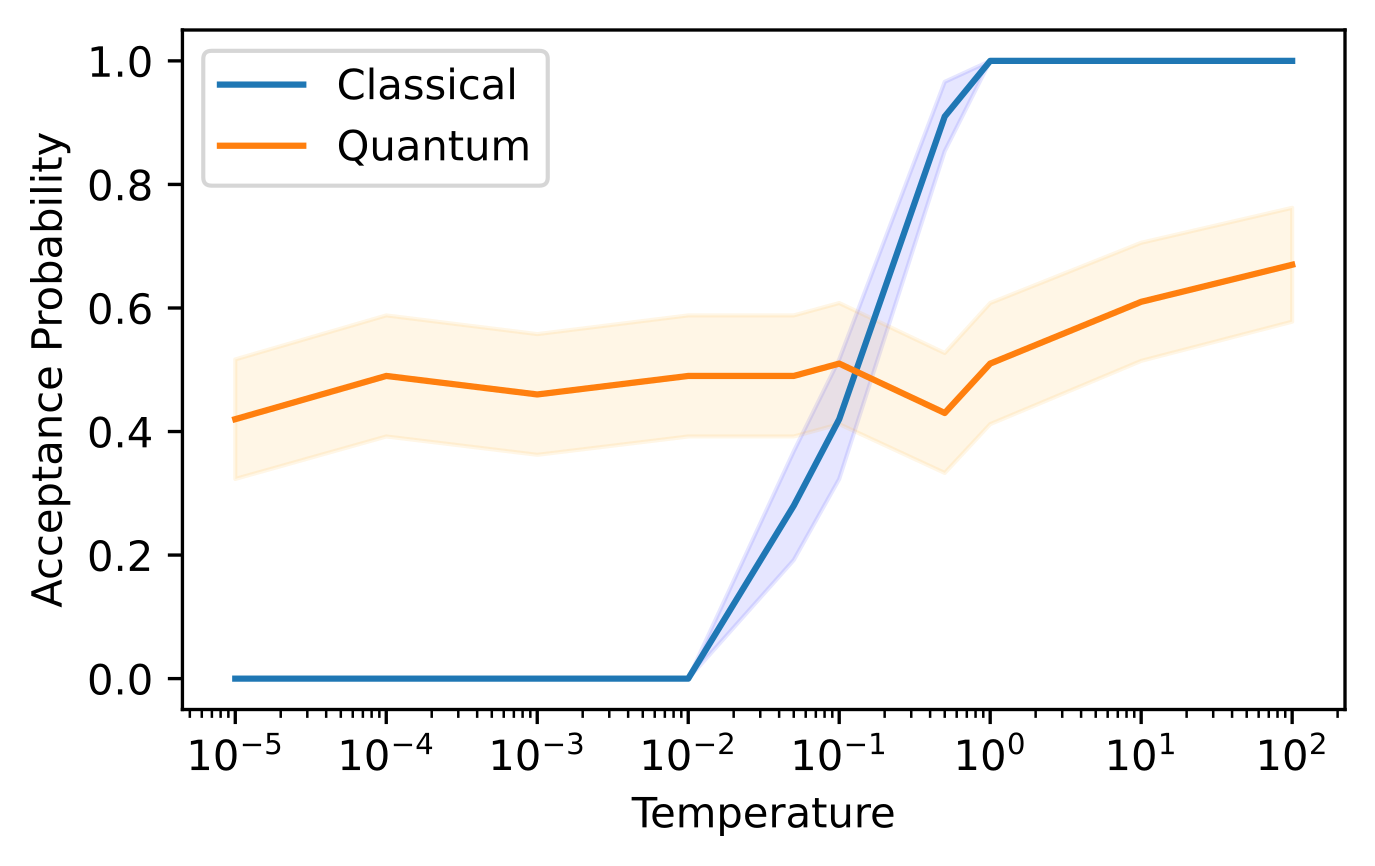}
    \caption{Acceptance probability comparison as a function of temperature for a 2D double well}
    \label{fig:accept}
\end{figure}

To explore QD-HMC further we focus on low temperature regimes. We can see the resilience of the proposals empirically when varying the temperature and observing the acceptance rates. Figure \ref{fig:accept} shows an example of this, in which we plot the acceptance probability of the quantum and classical proposals for 100 different random points for a 2D double well function. The hyperparameters used are the same as the $T = 0.1$ double well above. As can be seen, the quantum proposals average around 50\% independent of temperature, whereas the classical proposals are strongly positively correlated with temperature. The classical acceptance rates could likely be increased with hyperparameter optimization for each temperature, however, this often results in infeasible step sizes at very low temperatures. This result demonstrates that the areas in which it might be best to probe for the advantage of QD-HMC are likely at low temperatures and this highlights the potential benefit of lessened dependence on properly tuned hyperparameters.

\section{Discussion \& Conclusion}

First, let us remark on the important distinctions between QDHMC and the Quantum Enhanced MCMC (QEMCMC) of \cite{layden2022quantum}. While we take inspiration and build upon their work, QDHMC operates on a different class of problems (general continuous optimization vs. discrete ising models) and which we show is analogous to performing HMC on a quantum computer. This allows for implementation on both continuous variable and discrete quantum computers. Additionally, continuous dynamics adds difficulty in analytically computing advantage bounds. Unlike in the discrete case, we cannot materialize the transition matrix for exact analysis. \citeauthor{orfi2024bounding} \cite{orfi2024bounding} demonstrated by analyzing the spectral gap of the transition matrix that the QEMCMC of \citeauthor{layden2022quantum} \cite{layden2022quantum} offers no speedup in the general worst case. Although this proof of no speedup does not directly apply to QDHMC, it does motivate our experiments to focus on areas in which QDHMC could provide practical advantages and to focus on the mechanisms of potential advantages. The potential for a polynomial speedup in low temperature regimes occurs as the quantum dynamics become a Grover search \cite{grover1996fast} with a narrow Gaussian as the target state \cite{verdon2019quantum}. In contrast to the low temperature evaluations above, as the wavefunction becomes delocalized during the evolution at high temperatures, the proposals find themselves in a highly nonlocal part of the landscape (i.e. traversing larger distances). Under high temperature conditions, we see high-weight updates, compared to simulated annealing or HMC. Consider, as an illustrative example, Figure \ref{fig:traj}. These show the trajectory HMC and QD-HMC algorithms take respectively for a high temperature ($T = 100$) optimization in a simple 1D double well. The QD-HMC updates are substantially larger than the classical due to the aforementioned delocalization. Note that this cannot be trivially solved by increasing the HMC step size or integration steps, as this results in convergence failures on this example. The coherent evolution and delocalization in position space enables tunnelling, since the potential and kinetic strengths are randomly sampled. 

\begin{figure}
    \centering
    \includegraphics[width=0.95\linewidth]{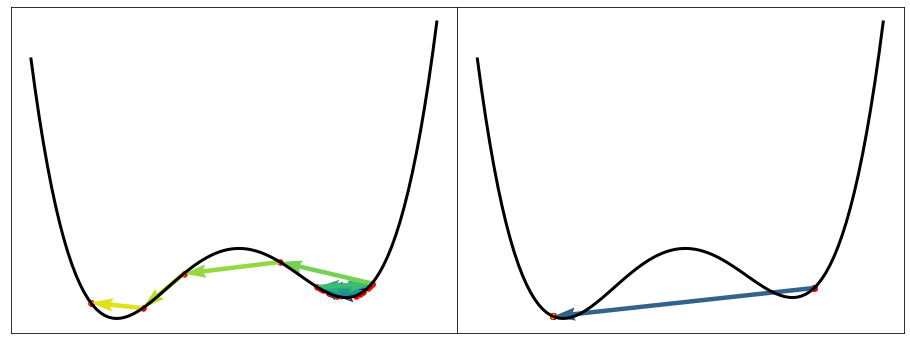}
    \caption{High temperature trajectory comparison between HMC (left) and QD-HMC (right)}
    \label{fig:traj}
\end{figure} 

This work begins to unlock a number of interesting future directions. There are improvements that could be made to this algorithm, e.g. the optimization of hyperparameters that defined the kicking magnitudes, how to tune all hyperparameters in certain optimization or probabilistic inference scenarios, etc. For example, one could tune the temperature synchronously with the temperature hyperparameters \cite{jordan_bang}. Another potential area of exploration would be the potential of population transfer for very low temperature updates \cite{smelyanskiy2020nonergodic}. Experimenting with real quantum hardware is another important future direction. Adapting any algorithm to hardware is a challenge as the simulation of the Hamiltonian induces many two qubit gates (which have substantially higher error rates) presenting an added challenge for this algorithm. However, the CV-QAOA subroutine, which is highly similar in structure to the transition kernel used in this work, was the subject of a recent hardware implementation \cite{enomoto2022continuous}. For digital simulation of continuous dynamics, accessing the higher energy states of qutrits (e.g. through superconducting transmon hardware \cite{cervera2022experimental}) may be helpful for hardware efficient emulation of qudit-based qumodes.

In this paper we demonstrated how to generalize HMC to quantum computing by leveraging coherent quantum dynamics. Building upon work done for quantum enhanced MCMC and CV-QAOA, we presented a new algorithm with promising machine learning use cases. This algorithm, QD-HMC, is theoretically outlined within the quantum mechanical and MCMC frameworks. We provide an empirical and theoretical analysis of our algorithm to demonstrate the potential advantages QD-HMC may offer. Given the prevalence and importance of classical HMC in a variety of applications, improvements on it are meaningful to substantial portions of the Bayesian, machine learning, and optimization communities. QD-HMC expands the potential applications of quantum hardware for all of these areas.

\section{Acknowledgments}

The authors thank Patrick Huembeli and Ian MacCormack for their valuable feedback on early versions of the manuscript.

\bibliography{Bibliography}


\appendix



\section{Appendix A: Heisenberg Picture Update Rules} \label{app:hp}

In the Heisenberg picture the wavefunction defines a time independent basis with time dependent operators, as opposed to the Schr\"odinger picture in which the wavefunction is a function of time and the operators are time independent. The resulting Heisenberg observables, $O^{(H)}$, can be described in relation to the Schr\"odinger observables, $O^{(S)}$, via

\begin{equation}
     O^{(H)}(t) = e^{i\hat{H}t/\hbar} O^{(S)} e^{-i\hat{H}t/\hbar}
\end{equation}

We can see the evolution of our operators at each layer of in the trotterization. The evolution of these operators in the Heisenberg picture can be represented via

\begin{equation}
    \dv{}{t} O^{(H)}(t) = \frac{i}{\hbar} [H, O^{(H)}(t)]
\end{equation}

With this background, we can now derive the updates in the right hand side of Equations \ref{eq:parshift} and \ref{eq:momshift} to see the evolution of the position and momentum operators (for a single timestep). Starting with the momentum operator we have the single (discrete trotter) update, i.e. when $t = 1$.

Since,

\begin{equation}
    \hat{H} = \eta\frac{\hat{\mathbf{p}}^{2}}{2} + \lambda f(\hat{\mathbf{x}})
\end{equation}

we have,

\spliteq{
    \dv{}{t} \mathbf{\hat{x}} &= \frac{i}{\hbar} \biggl [  \eta\frac{\hat{\mathbf{p}}^{2}}{2} + \lambda f(\hat{\mathbf{x}}), \mathbf{\hat{x}} \biggr ] \\
    &= \frac{i}{\hbar} \biggl ( \biggl ( \eta\frac{\hat{\mathbf{p}}^{2}}{2} + \lambda f(\hat{\mathbf{x}}) \biggr ) \hat{\mathbf{x}} - \hat{\mathbf{x}} \biggl ( \eta\frac{\hat{\mathbf{p}}^{2}}{2} + \lambda f(\hat{\mathbf{x}}) \biggr ) \biggr ) \\
    &= \frac{i}{\hbar} \biggl ( \eta\frac{\hat{\mathbf{p}}^{2}}{2} \hat{\mathbf{x}} + \lambda f(\hat{\mathbf{x}}) \hat{\mathbf{x}} - \hat{\mathbf{x}} \eta\frac{\hat{\mathbf{p}}^{2}}{2} -  \hat{\mathbf{x}} \lambda f(\hat{\mathbf{x}}) \biggr ) 
}

Using a temporary wavefunction for clarity, and knowing that $f(\hat{\mathbf{x}})$ and $\hat{\mathbf{x}}$ commute, we can see

\spliteq{
    & \frac{i}{\hbar} \biggl ( \eta\frac{\hat{\mathbf{p}}^{2}}{2} \hat{\mathbf{x}} + \lambda f(\hat{\mathbf{x}}) \hat{\mathbf{x}} - \hat{\mathbf{x}} \eta\frac{\hat{\mathbf{p}}^{2}}{2} -  \hat{\mathbf{x}} \lambda f(\hat{\mathbf{x}}) \biggr ) \psi \\
    &= \frac{i}{\hbar} \biggl ( \eta\frac{\hat{\mathbf{p}}^{2}}{2} \hat{\mathbf{x}} \psi - \hat{\mathbf{x}} \eta\frac{\hat{\mathbf{p}}^{2}}{2} \psi  \biggr ) \\
    &= \frac{i}{2 \hbar} \eta \Bigl [ \hat{\mathbf{p}}^{2} , \hat{\mathbf{x}}  \Bigr ] \\
    &= \frac{i}{2 \hbar} \eta \biggl ( \Bigl [ \hat{\mathbf{p}} , \hat{\mathbf{x}}  \Bigr ]  \hat{\mathbf{p}} +  \hat{\mathbf{p}} \Bigl [  \hat{\mathbf{p}} , \hat{\mathbf{x}} \Bigr ] \biggr ) \\
    &= \frac{i}{2 \hbar} \eta \Bigl ( -i\hbar \hat{\mathbf{p}}  - i\hbar \hat{\mathbf{p}}  \Bigr ) \\
    &= \eta \hat{\mathbf{p}}  
}

Thus we can see that the discrete trotter step enacts a change on $\hat{\mathbf{x}}$ via $ \eta \hat{\mathbf{p}}$, thus recovering the update in Equation \ref{eq:parshift}. 

Now let us walk through the same steps, but for momentum. 

We have,

\spliteq{
    \dv{}{t} \mathbf{\hat{p}} &= \frac{i}{\hbar} \biggl [  \eta\frac{\hat{\mathbf{p}}^{2}}{2} + \lambda f(\hat{\mathbf{x}}), \mathbf{\hat{p}} \biggr ] \\
    &= \frac{i}{\hbar} \biggl ( \biggl ( \eta\frac{\hat{\mathbf{p}}^{2}}{2} + \lambda f(\hat{\mathbf{x}}) \biggr ) \hat{\mathbf{p}} - \hat{\mathbf{p}} \biggl ( \eta\frac{\hat{\mathbf{p}}^{2}}{2} + \lambda f(\hat{\mathbf{x}}) \biggr ) \biggr ) \\
    &= \frac{i}{\hbar} \biggl ( \eta\frac{\hat{\mathbf{p}}^{3}}{2}  + \lambda f(\hat{\mathbf{x}}) \hat{\mathbf{p}} -  \eta\frac{\hat{\mathbf{p}}^{3}}{2} -  \hat{\mathbf{p}} \lambda f(\hat{\mathbf{x}}) \biggr ) 
}

Adding in a wavefunction once again, we can see

\spliteq{
    & \frac{i}{\hbar} \biggl ( \eta\frac{\hat{\mathbf{p}}^{3}}{2}  + \lambda f(\hat{\mathbf{x}}) \hat{\mathbf{p}} -  \eta\frac{\hat{\mathbf{p}}^{3}}{2} -  \hat{\mathbf{p}} \lambda f(\hat{\mathbf{x}}) \biggr )  \psi \\
    &= \frac{i}{\hbar} \lambda \Bigl ( f(\hat{\mathbf{x}}) \hat{\mathbf{p}} \psi -  \hat{\mathbf{p}} f(\hat{\mathbf{x}}) \psi \Bigr ) \\
    &= \lambda \Bigl ( f(\hat{\mathbf{x}}) \nabla \psi -  \nabla ( f(\hat{\mathbf{x}}) \psi) \Bigr ) \\
    &= \lambda \Bigl ( f(\hat{\mathbf{x}}) \nabla \psi - f(\hat{\mathbf{x}}) \nabla \psi - \psi \nabla f(\hat{\mathbf{x}}) \Bigr ) \\
    &= -\lambda \nabla f(\hat{\mathbf{x}})
}

Thus we can see that the discrete trotter step enacts a change on $\hat{\mathbf{p}}$ via $ -\lambda \nabla f(\hat{\mathbf{x}}) $, thus recovering Equation \ref{eq:momshift}.

\section{Appendix B: QD-HMC and detailed balance}\label{app:balance}
In this section we will prove that QD-HMC, as outlined in Algorithm \ref{alg:app_algo}, obeys the detailed balance condition, as it is essential to showing that the sampled distribution converges to the true target Boltzmann distribution in the asymptotic limit.

\begin{algorithm}
\caption{QD-HMC Algorithm}\label{alg:app_algo}
\begin{algorithmic}
    \State $1.$ Prepare bitstring $\ket{x}$
    \State $2.$ Where $f(\x)$ is the problem function, define $$\H_{\eta\la} = \eta\frac{\p^{2}}{2} + \la f(\x)$$ with chosen hyperparameters $\eta, \la$.
    \State $3.$ Via random Trotterization, apply $$\U_{\eta\la} \equiv e^{-i \H_{\eta\la} t}$$
    \State $4.$ Optionally, flip the momentum via Algorithm \ref{alg:flip}
    \State $4.$ Sample bitstring measurement $$y \sim \lvert\bra{y}\U_{\eta\la}\ket{x}\rvert^{2}$$
    \State $5.$ Accept with probability $$A(y|x)=\min\{1,e^{-[f(y)-f(x)]}\}$$
\end{algorithmic}
\end{algorithm}

\begin{claim} \label{eq:1}
\beq
    \lvert\bra{y}\U_{\eta\la}\ket{x}\rvert^{2} = \lvert\bra{x}\U_{\eta\la}\ket{y}\rvert^{2}
    \qquad
    \forall x,y,\eta,\la
\eeq
\end{claim}

\begin{proof}
\begin{note}
    $(e^{\A})\T = e^{\A\T}$
\end{note}

(B1) holds iff $\U_{\eta\la} = \h{U}_{\eta\la}\T$

This equality is equivalent to:

\spliteq{
    & \Longleftrightarrow 
    \left[ \eta \p^{2} + \la f(\h{x}) \right]\T = \eta \p^{2} + \la f(\x)
    \\ 
    &\Longleftrightarrow 
    \left(\p^{2}\right)\T = \p^{2}
}

\begin{note}
    $\p=\F\x\F^{\dag}$ and $\F\T=\F$.
\end{note}
Therefore,
\spliteq{
    \p\T &= \left( \F^{\dag} \right)\T\x\F\T = \F^{\dag}\x\F
    \\ &= \left(\F^{\dag}\right)^{2} \p\F^{2} = -\p
    }

We can make this transition invertible by adding a Momentum Flip as seen in Algorithm \ref{alg:flip}. \\ 

\begin{algorithm}
\caption{Momentum Flip}\label{alg:flip}
\begin{algorithmic}
    \State $1.$ Momentum flip: 
    $$\M \equiv \F_{c}^{\dag}\X_{j=n}\F_{c}, \quad \M^{\dag} = \M$$
    \State $2.$ Convert to momentum space
    \State $3.$ Flip most significant bit
    \State $4.$ Convert back
\end{algorithmic}
\end{algorithm}

\twocolumngrid 

\beq
    \U_{\eta\la} \mapsto \M\U_{\eta\la} \mid \M\U_{\eta\la}\M = \U_{\eta\la}\T
\eeq

Since,

\spliteq{
    \M\M\T &= \F_{c}^{\dag}\X_{n}\F_{c}\F_{c}\X_{n}\F_{c}^{\dag} \\
    &= -\F_{c}^{\dag}\F_{c}\F_{c}\F_{c}^{\dag} \\
    &= -I
}

It follows that,

\spliteq{
    \left( \M\U \right)\T &= \U\T\M\T
    \\ &= \M\U\M\M\T
    \\ &= \M\U
}
Thus,
\beq
    \lvert\bra{y}\U_{\eta\la}\ket{x}\rvert^{2} = \lvert\bra{x}\U_{\eta\la}\ket{y}\rvert^{2}
    \qquad
    \forall x,y,\eta,\la
\eeq
\end{proof}

\section{Appendix C: Further Autocorrelation Estimates}\label{app:ac}

Figure \ref{fig:tau_g} shows the estimates of $\tau$ for different 2D problems.

\begin{figure}[H]
    \centering
    \includegraphics[width=0.95\linewidth]{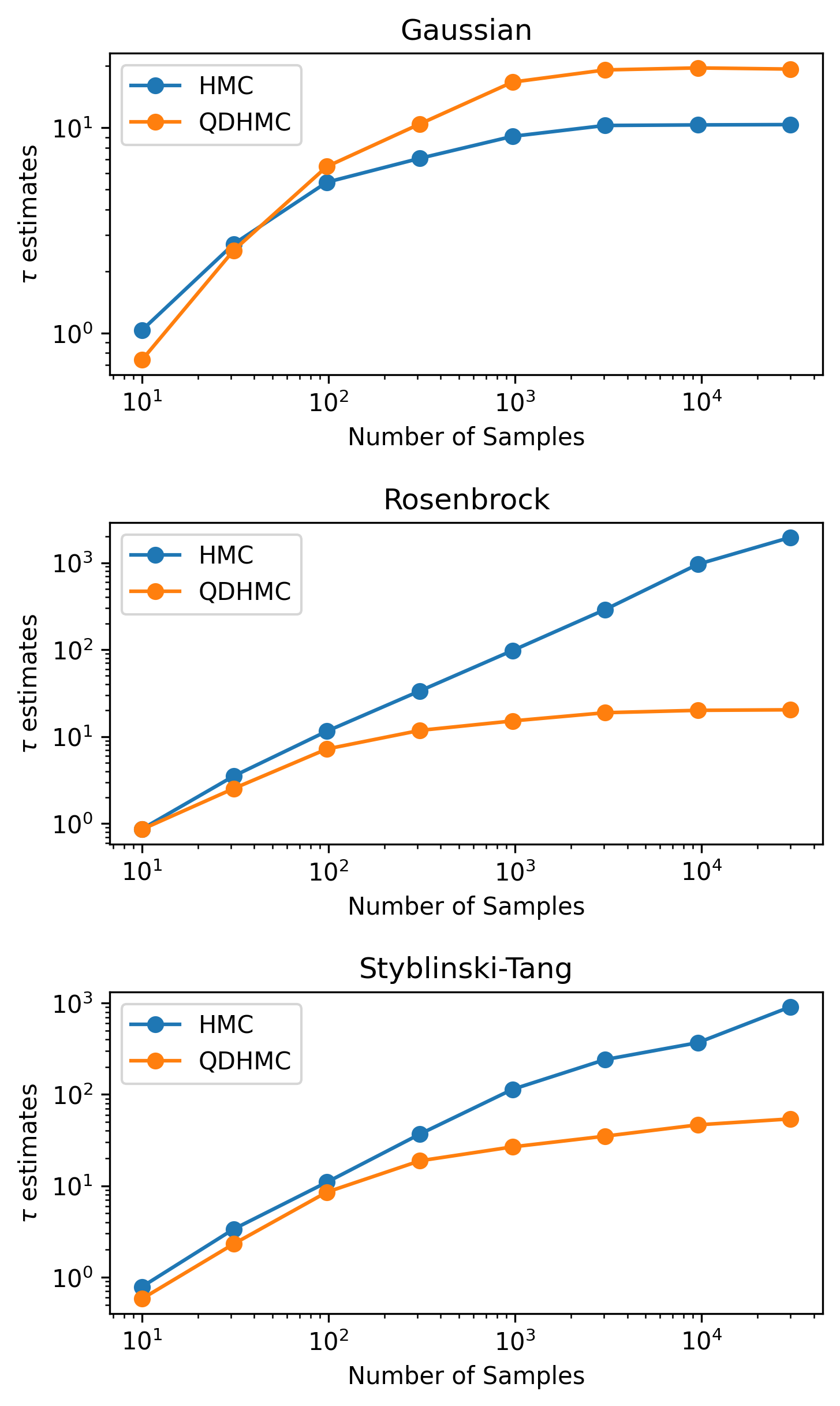}
    \caption{Estimation of autocorrelation time $\tau$ as a function of number of samples}
    \label{fig:tau_g}
\end{figure}

\end{document}